\documentclass{IEEEtran}

\usepackage{amsmath,amssymb} \interdisplaylinepenalty=2500
\usepackage{cite}
\usepackage{graphicx} %\graphicspath{{./figs/}}

\newtheorem{definition}{Definition}
\newtheorem{theorem}{Theorem}
\newtheorem{proposition}[theorem]{Proposition}

\newtheorem{lemma}[theorem]{Lemma}

\newtheorem{myexample}{Example}
\newenvironment{example}{\myexample}{\qed\endmyexample}
\newenvironment{remark}{\textit{Remark: }}{}
\def\qed{\endIEEEproof}

\renewcommand{\dim}{\operatorname{\sf dim}\hspace{0.1em}}
\DeclareMathOperator{\rank}{\sf rank\hspace{0.1em}}
\DeclareMathOperator{\wt}{\sf wt\hspace{0.1em}}
\newcommand{\linspan}[1]{\left< #1 \right>}

\newcommand{\rowsp}{\linspan}

\newcommand{\Fq}{\mathbb{F}_q}
\newcommand{\Fqm}{\mathbb{F}_{q^m}}
\newcommand{\Fqn}{\mathbb{F}_{q^n}}

\newcommand{\mat}[1]{\begin{bmatrix} #1 \end{bmatrix}}

\newcommand{\dr}{d_{\rm R}}

\newcommand{\calA}{\mathcal{A}}

\newcommand{\calC}{\mathcal{C}}

\newcommand{\calE}{\mathcal{E}}

\newcommand{\calI}{\mathcal{I}}

\newcommand{\calP}{\mathcal{P}}

\newcommand{\calR}{\mathcal{R}}
\newcommand{\calS}{\mathcal{S}}

\newcommand{\calW}{\mathcal{W}}
\newcommand{\calX}{\mathcal{X}}
\newcommand{\calY}{\mathcal{Y}}
\newcommand{\calZ}{\mathcal{Z}}

\newcommand{\frakR}{\mathfrak{R}}
\newcommand{\dec}{\textsf{D}}
\newcommand{\enc}{\textsf{E}}
\newcommand{\Gin}{G_{{\rm sec}}}

\title{Universal Secure Network Coding\\ via Rank-Metric Codes}

\author{Danilo Silva and Frank R. Kschischang%
\thanks{The work of D. Silva was supported in part by CAPES Foundation (Brazil) and by FAPESP (Brazil) grant 2009/15771-7. The material in this paper was presented in part at the IEEE International Symposium on Information Theory, Toronto, Canada, July 2008, and in part at the IEEE International Symposium on Information Theory, Austin, TX, June 2010.}
\thanks{D. Silva was with The Edward S. Rogers Sr. Department of Electrical and Computer Engineering, University of Toronto, Toronto, ON M5S 3G4, Canada. He is now with the School of Electric and Computer Engineering, State University of Campinas, Campinas, SP 13083-970, Brazil (e-mail: danilo@decom.fee.unicamp.br).}
\thanks{F. R. Kschischang is with The Edward S. Rogers Sr. Department of Electrical and Computer Engineering, University of Toronto, Toronto, ON M5S 3G4, Canada (e-mail: frank@comm.utoronto.ca).}
}

\begin{document}
\maketitle
\thispagestyle{empty}

\begin{abstract}

The problem of securing a network coding communication system against an eavesdropper adversary is considered. The network implements linear network coding to deliver $n$ packets from source to each receiver, and the adversary can eavesdrop on $\mu$ arbitrarily chosen links. The objective is to provide reliable communication to all receivers, while guaranteeing that the source information remains information-theoretically secure from the adversary. A coding scheme is proposed that can achieve the maximum possible rate of $n-\mu$ packets. The scheme, which is based on rank-metric codes, has the distinctive property of being \emph{universal}: it can be applied on top of any communication network without requiring knowledge of or any modifications on the underlying network code. The only requirement of the scheme is that the packet length be at least $n$, which is shown to be strictly necessary for universal communication at the maximum rate.
A further scenario is considered where the adversary is allowed not only to eavesdrop but also to inject up to $t$ erroneous packets into the network, and the network may suffer from a rank deficiency of at most $\rho$. In this case, the proposed scheme can be extended to achieve the rate of $n-\rho-2t-\mu$ packets. This rate is shown to be optimal under the assumption of zero-error communication.

\end{abstract}

\section{Introduction}
\label{sec:introduction}

The paradigm of network coding \cite{Ahlswede++2000,Li++2003,Koetter.Medard2003} has provided a rich source of new problems that generalize traditional problems in communications. One such problem, introduced in \cite{Cai.Yeung2002:Secure} by Cai and Yeung, is that of securing a multicast network against an eavesdropper adversary.

Formally, consider a multicast network with unit capacity edges implementing linear network coding over a finite field $\Fq$. It is assumed that each link in the network carries a packet consisting of $m$ symbols in $\Fq$ and that the network is capable of reliably transporting $n$ packets from the source to each destination. Now, suppose there is an eavesdropper that can listen to transmissions on $\mu$ arbitrarily chosen links\footnote{We consider a model where network links rather than nodes are eavesdropped; eavesdropping on a node is equivalent to eavesdropping on all links incoming to it.} of the network. The secure network coding problem is to design an outer code (and possibly also the underlying network code) such that a message can be communicated to each receiver without leaking any information to the eavesdropper (i.e., security in the information-theoretic sense).

The work of Cai and Yeung \cite{Cai.Yeung2002:Secure} shows that the maximum achievable rate (i.e., the \emph{secrecy capacity}) for this problem is given by $n-\mu$ packets, achievable if the field size $q$ is sufficiently large. They presented a construction of an outer code that achieves this capacity provided that $q \geq \binom{|\calE|}{\mu}$. Their construction takes $O(q)$ steps and requires that the outer code meet certain security conditions imposed by the underlying network code. Later, Feldman et al. \cite{Feldman++2004:CapacitySecure} showed that, by slightly reducing the rate, it is possible to efficiently construct an outer code that is secure with high probability using a much smaller field size. On the other hand, they also showed that, under the assumption of a scalar linear outer code, there are instances of the problem where a very large field size is strictly necessary to achieve capacity.

More recently, Rouayheb and Soljanin \cite{Rouayheb.Soljanin2007} showed that the secure network coding problem can be regarded as a network generalization of the Ozarow-Wyner wiretap channel of type II \cite{Ozarow.Wyner1984,Ozarow.Wyner1985}. Their observation provides an important connection with a classical problem in information theory and yields a much more transparent framework for dealing with network coding security. In particular, they show that the same technique used to achieve capacity of the wiretap channel II---a coset coding scheme based on a linear MDS code---can also provide security for a wiretap network. Unfortunately, in their approach, the network code has to be modified to satisfy certain constraints imposed by the outer code.

Note that, in all the previous works, either the network code has to be modified to provide security \cite{Rouayheb.Soljanin2007}, or the outer code has to be designed based on the specific network code used \cite{Cai.Yeung2002:Secure,Feldman++2004:CapacitySecure}. In all cases, the network code must be known beforehand, and the field size required is significantly larger than the minimum required for conventional multicasting.

The present paper is motivated by Rouayheb and Soljanin's formulation of a wiretap network and builds on their results. Our first main contribution is a coset coding scheme that neither imposes any constraints on, nor requires any knowledge of, the underlying network code. In other words, for any linear network code that is feasible for multicast, secure communication at the maximum possible rate can be achieved with a fixed outer code. In particular, the field size can be chosen as the minimum required for multicasting. In this paper, such network-code-independent schemes are called \emph{universal}. An important consequence of our result is that, if universal schemes are assumed, then the problem of information transport (i.e., designing a feasible network code) and the problem of security against an eavesdropper can be completely separated from each other. In particular, universal schemes can be seamlessly integrated with random network coding.

The essence of our approach is to use a vector linear outer code. More precisely, we regard packets as elements of an extension field $\Fqm$, and use an outer code that is linear over $\Fqm$. Taking advantage of this extension field, we can then replace the linear MDS code in Ozarow-Wyner coset coding scheme by a linear maximum-rank-distance (MRD) code, which is essentially a linear code over $\Fqm$ that is optimal for the rank metric. Codes in the rank metric were studied by a number of authors \cite{Gabidulin1985,Roth1991:MaximumRankArrayCodes,Gadouleau.Yan2008} and have been proposed for error control in random network coding \cite{Silva.Kschischang2007:ISIT,Silva++2008}.
Here, we show that, since the channel to the eavesdropper is a linear transformation channel (rather than an erasure channel), rank-metric codes are naturally suitable to the problem (as opposed to classical codes designed for the Hamming metric).

Another main contribution of this paper is the design of universal schemes that can provide both security and protection against errors. More precisely, we assume that the adversary is able not only to eavesdrop on $\mu$ arbitrarily chosen links, but also to inject $t$ erroneous packets anywhere in the network. We also assume that the network may suffer from a rank deficiency of at most $\rho$ packets. Previous work on this topic includes \cite{Ngai.Yang2007:DeterministicSEC} and \cite{Ngai.Yeung2009:Secure-Error-Correcting}, which propose secure-error-correcting schemes achieving a rate $n-\rho-2t-\mu$. However, these schemes are not universal and suffer from the same issues as the Cai-Yeung scheme discussed above.

Note that the naive approach to this problem would be simply to concatenate a secrecy encoder with an error control encoder. However, the security of such a scheme is not guaranteed because the error control encoder can potentially ``undo'' part of the secrecy encoding. On the other hand, if secrecy encoding is applied after error control encoding, then, due to the same reason, the concatenated scheme is not guaranteed to provide error control.

Our approach to this problem is to design a single scheme that simultaneously provides security and error control, by leveraging the corresponding properties of rank-metric codes. Our proposed scheme is universal and achieves the rate of $n - \rho - 2t - \mu$ packets. We show that this rate is indeed optimal, under the assumption of zero-error communication\footnote{If this assumption is relaxed to vanishingly small error probability, then higher rates may be achieved in some cases. See \cite{Yao++2010:NetCod-JammingEavesdropping}.}. This result (whose proof allows arbitrary packet lengths) generalizes a similar bound in \cite{Ngai.Yeung2009:Secure-Error-Correcting} that assumed packet length $m=1$ (i.e., a scalar linear outer code).

All the universal schemes proposed in this paper have a single limitation: the packet length must satisfy $m \geq n$. While this requirement is usually easily satisfied in the practice of random network coding (see, e.g., \cite{Chou++2003}), we show that it is also strictly necessary for universal communication. In other words, universal schemes that provide security and/or error control at the maximum rate do not exist if $m < n$. Thus, our proposed schemes are optimal also in the sense of requiring the smallest packet size among all universal schemes.

The remainder of the paper is organized as follows. Section~\ref{sec:preliminaries} presents a brief review of rank-metric codes and the basic model of linear network coding. In Section~\ref{sec:problem-formulation}, we formulate the problem of universal secure and reliable communication over a wiretap network, following the basic setup of the wiretap channel. In Section~\ref{sec:error-correction}, we start by addressing the special case where only error control is required. We prove a few auxiliary results that extend the results of \cite{Silva.Kschischang2009:Metrics}. Then, in Section~\ref{sec:noiseless}, we address the special case where only security is required. The complete scenario of both security and error control is addressed in Section~\ref{sec:noisy}. In Section~\ref{sec:practical-considerations}, we discuss the practical application of our proposed schemes, and show that they can be implemented in a convenient and very efficient manner. Finally, Section~\ref{sec:conclusion} presents our conclusions. 

Previous versions of this work appeared in \cite{Silva.Kschischang2008:ISIT,Silva2009(PhD),Silva.Kschischang2010:ISIT}.

%In practice, these parameters provide an interface between two layers: the underlying network coding system, whose objective is to saturate the capacity with highest computational efficiency and minimal overhead; and the end-to-end coding scheme, which can provide security and reliability.

\section{Preliminaries}
\label{sec:preliminaries}

\subsection{Notation}
\label{ssec:notation}

Let $\Fq^{n \times m}$ denote the set of all $n \times m$ matrices over $\Fq$, and set $\Fq^n \triangleq \Fq^{n \times 1}$ (i.e., the elements of $\Fq^n$ are always seen as \emph{column} vectors).
For $M \in \Fq^{n \times m}$ and $\calS \subseteq \{1,\ldots,n\}$, let $M_\calS$ denotes the submatrix of $M$ consisting of the rows indexed by $\calS$. Let $\linspan{M}$ denote the row space of matrix $M$.

\subsection{Rank-Metric Codes}
\label{ssec:rank-metric-codes}

A \emph{matrix code} is a nonempty set of matrices. The \emph{rank distance} between matrices $X,Y \in \Fq^{n \times m}$ is defined as
\begin{equation}\nonumber
  \dr(X,Y) \triangleq \rank(Y - X).
\end{equation}
As observed in \cite{Gabidulin1985,Roth1991:MaximumRankArrayCodes}, the rank distance is indeed a \emph{metric}. The \emph{minimum rank distance} of a matrix code $\calC \subseteq \Fq^{n \times m}$, denoted $\dr(\mathcal{C})$, is the minimum rank distance among all pairs of distinct codewords of $\mathcal{C}$.

Let $\Fqm$ be a degree $m$ extension of the finite field $\Fq$. Recall that $\Fqm$ is also a vector space over $\Fq$. Let $\phi_m: \Fqm \to \Fq^{1 \times m}$ be a vector space isomorphism. More concretely, $\phi_m$ expands an element of $\Fqm$ as a row vector over $\Fq$ according to some fixed basis for $\Fqm$ over $\Fq$. Similarly, for all $n,\ell$, let $\phi_m^{(n \times \ell)}: \Fqm^{n \times \ell} \to \Fq^{n \times \ell m}$ be the isomorphism defined by applying $\phi_m$ entry-wise, i.e.,
\begin{equation}\nonumber
  \phi_m^{(n \times \ell)}(X) = \mat{\phi_m(X_{11}) & \cdots & \phi_m(X_{1\ell}) \\ \vdots & & \vdots \\ \phi_m(X_{n1}) & \cdots & \phi_m(X_{n \ell}))}.
\end{equation}
We will remove the superscript from $\phi_m^{(n \times \ell)}$ when the dimensions of the argument are clear from the context. The rank distance between vectors $X,Y \in \Fqm^n$ and the minimum rank distance of a block code $\calC \subseteq \Fqm^n$ are defined, respectively, as $\dr(X,Y) \triangleq \dr(\phi_m(X),\phi_m(Y))$ and $\dr(\calC) \triangleq \dr(\phi_m(\calC))$.

The size of a matrix code (or of a block code over $\Fqm$) is bounded by the Singleton bound for the rank metric, which states that every $\calC \subseteq \Fq^{n \times m}$ with minimum rank distance $d$ must satisfy
\begin{equation}\label{eq:singleton-bound}
|\mathcal{C}| \leq q^{\max\{n,m\} (\min\{n,m\} - d + 1)}.
\end{equation}
Codes that achieve this bound are called \emph{maximum-rank-distance} (MRD) codes and they are known to exist for all choices of parameters $q$, $n$, $m$ and $d \leq \min\{n,m\}$ \cite{Gabidulin1985}.

For the case of an $[n,k]$ \emph{linear} block code over $\Fqm$ with minimum rank distance $d$, the Singleton bound (\ref{eq:singleton-bound}) becomes
\begin{equation}\label{eq:singleton-linear}
  d \leq \min\left\{1,\frac{m}{n}\right\} (n-k) + 1.
\end{equation}
Note that, for $m \geq n$, (\ref{eq:singleton-linear}) coincides with the classical Singleton bound for the Hamming metric. Indeed, when $m \geq n$, every MRD code is also MDS.

Note that, differently from classical coding theory, block codes are represented in this paper using \emph{column} vectors. However, to avoid confusion, the generator and parity-check matrices of a linear code will always be given in the standard orientation. Thus, if $G \in \Fqm^{k \times n}$ and $H \in \Fqm^{(n-k) \times n}$ are, respectively, the generator and parity-check matrices of an $[n,k]$ linear code $\calC \subseteq \Fqm^n$, then $\calC = \{G^T u: u \in \Fqm^k\} = \{x \in \Fqm^n: Hx = 0\}$.

We now describe an important family of rank-metric codes proposed by Gabidulin \cite{Gabidulin1985}. Assume $m \geq n$.
%For convenience, let $[i]$ denote $q^i$.
A \emph{Gabidulin code} is an $[n,k]$ linear code over $\Fqm$ defined
% by the parity-check matrix
%%$H = \mat{h_{j}^{[i]}}$, $0 \leq i \leq n-k-1$, $0 \leq j \leq n-1$,
%\begin{equation}\label{eq:parity-check-gabidulin}
%H = \mat{
%  h_0^{[0]} & h_1^{[0]} & \cdots & h_{n-1}^{[0]} \\
%  h_0^{[1]} & h_1^{[1]} & \cdots & h_{n-1}^{[1]} \\
%  \vdots & \vdots  & \ddots & \vdots \\
%  h_0^{[n-k-1]} & h_1^{[n-k-1]} & \cdots & h_{n-1}^{[n-k-1]}
%  }
%\end{equation}
%where the elements $h_0,\ldots,h_{n-1} \in \Fqm$ are linearly independent
%over $\Fq$.
by the generator matrix
\begin{equation}\label{eq:generator-gabidulin}
G = \mat{
  g_0^{q^{0}} & g_1^{q^{0}} & \cdots & g_{n-1}^{q^{0}} \\
  g_0^{q^{1}} & g_1^{q^{1}} & \cdots & g_{n-1}^{q^{1}} \\
  \vdots & \vdots  & \ddots & \vdots \\
  g_0^{q^{k-1}} & g_1^{q^{k-1}} & \cdots & g_{n-1}^{q^{k-1}}
  }
\end{equation}
where the elements $g_0,\ldots,g_{n-1} \in \Fqm$ are linearly independent over $\Fq$. It is shown in \cite{Gabidulin1985} that the minimum rank distance of a Gabidulin code is $d = n-k+1$, so the code is MRD.

\subsection{Linear Network Coding}
\label{ssec:linear-network-coding}

A linear network coding system is described as follows. Consider a communication network represented by a directed multigraph with unit capacity edges, a single source node, and multiple destination nodes. Each link in the network is assumed to transport, free of errors, a \emph{packet} consisting of $m$ symbols from the finite field $\Fq$ (that is, a vector in $\Fq^{1 \times m}$). At every network use, the source node produces $n$ packets, represented as the rows of a matrix $X \in \Fq^{n \times m}$, and transmits evidence about these packets over the network. More precisely, for each of its outgoing links, the source node transmits a packet that is some $\Fq$-linear combination of the rows of $X$. Each of the remaining nodes behaves similarly, computing its outgoing packets as $\Fq$-linear combinations of its incoming packets. It follows that, for every link $e$, the packet $P_e$ transmitted over $e$ can be expressed (uniquely) as a linear combination of the rows of $X$, say $P_e = c_e X$. The coefficient vector $c_e \in \Fq^{1 \times n}$ is called the \emph{(global) coding vector} of $P_e$. Let $\calE$ denote the set of all network links, ordered according to some fixed ordering. A \emph{(global) coding matrix} $C \in \Fq^{|\calE| \times n}$ is defined such that, for all $e \in \calE$, $c_e$ is the row of $C$ indexed by $e$.

For analytical purposes, a receiver can be specified, without loss of generality, by the set of incoming links of the corresponding destination node. Let $\mathfrak{R}$ denote the collection of all receivers. Note that $\mathfrak{R}$ is a subset of the powerset of $\calE$. For $\calR \in \frakR$, let $Y(\calR) \in \Fq^{|\calR| \times m}$ denote the matrix whose rows are the packets received by receiver $\calR$. Then
\begin{equation}\nonumber
  Y(\calR) = C_\calR X.
\end{equation}
The network code is said to be feasible for a receiver $\calR$ if $\rank C_\calR = n$, otherwise it is rank-deficient. The \emph{rank deficiency} of a network code is defined as
\begin{equation}\nonumber
   \rho = n - \min_{\calR \in \mathfrak{R}} \rank C_\calR
\end{equation}
i.e., it is the maximum column-rank deficiency of $C_\calR$ among all receivers. Since, in a network coding context, rank deficiency is analogous to packet loss, a rank deficiency of $\rho$ may also be referred to as \emph{$\rho$ packet erasures}.

The system described above is referred to as an \emph{$(n \times m)_q$ linear coded network}. We may also call it an $(n \times m,\, k)_q$ linear coded network if its rank deficiency is $\rho = n-k$.

We can extend the above model to incorporate packet errors. More precisely, we assume that each packet transmitted on a link may be subject to the addition of an error packet before reception by the corresponding node. This is useful to model both internal adversaries (malicious nodes that inject erroneous packets) as well as external adversaries (unauthorized transmitters that intentionally create interference with the transmitted signals). Similarly as above, this communication model can be described more concisely using a matrix framework. Suppose the packet transmitted on link $j$ changes from $P_j$ to $P_j'$. Then, due to linearity of the network, the packet transmitted on link $i$ changes from $P_i$ to $P_i' = P_i + F_{i,j} (P_j' - P_j)$, for some $F_{i,j} \in \Fq$. Let $F \in \Fq^{|\calE| \times |\calE|}$ be the matrix whose $(i,j)$ entry is $F_{i,j}$. Then the matrix received by receiver $\calR$ is given by
\begin{equation}\nonumber
  Y(\calR) = C_\calR X + F_\calR Z
\end{equation}
where $Z \in \Fq^{|\calE| \times m}$ is the matrix corresponding to the error packets injected on all links.

Note that the above model is applicable even if the network contains cycles and/or delays.

\section{Problem Formulation}
\label{sec:problem-formulation}

We start by describing a generic wiretap channel. Let $\calS$, $\calX$, $\calY$ and $\calW$ be sets. A transmitter wishes to communicate a message $S \in \calS$ reliably to a receiver but secretly from an eavesdropper. There is a channel among the three parties that takes $X \in \calX$ from the transmitter and delivers $Y \in \calY$ to the receiver and $W \in \calW$ to the eavesdropper. The channel is specified by some distribution $P(Y,W|X)$. The transmitter generates $X$ by a stochastic encoding of $S$, according to some distribution $P(X|S)$. Upon reception of $Y$, the receiver makes a guess that the transmitted message is $\dec(Y)$, according to some decoding function $\dec \colon \calY \to \calS$. The eavesdropper, on the other hand, attempts to obtain information about $S$ based on the observation $W$. Together, the encoder $P(X|S)$ and the decoder $\dec(\cdot)$ specify the coding scheme. Note that $S$ and $(Y,W)$ are assumed to be conditionally independent given $X$.

\subsection{Communication Requirements}
\label{ssec:requirements}

We now describe the requirements that a coding scheme must satisfy, for the purposes of this paper.

\subsubsection{Zero-error communication}

For all $x \in \calX$, let the \emph{fan-out set}
\begin{equation}\nonumber
  \calY_x \triangleq \{y \in \calY \colon P(y|x) > 0\}
\end{equation}
denote the set of all channel outputs that could \emph{possibly} occur when the channel input is $x$. Similarly, for all $s \in \calS$, let
\begin{align}
\calX_s &\triangleq \{x \in \calX \colon P(x|s) > 0\} \nonumber \\
\calY_{(s)} &\triangleq \{y \in \calY \colon P(y|s) > 0\} = \bigcup_{x \in \calX_s}\, \calY_x. \label{eq:fan-out-overall}
\end{align}
The scheme is said to be \emph{zero-error} if
\begin{equation}\label{eq:condition-zero-error}
  \dec(y) = s, \quad \text{for all $y \in \calY_{(s)}$ and all $s \in \calS$}
\end{equation}
that is, the receiver can always uniquely determine the message\footnote{Our focus on zero-error communication is motivated by the goal of guaranteeing reliability in the presence of adversarial jammers. Since an adversary will attempt to disrupt communication whenever such a possibility \emph{exists}, requiring zero-error (``foolproof'') communication appropriately captures the worst-case nature of the problem.}. We may also refer simply to the encoder $P(X|S)$ and consider it zero-error if there \emph{exists} some decoding function $\dec(\cdot)$ satisfying (\ref{eq:condition-zero-error}). It is easy to see that an encoder is zero-error if and only if the sets $\calY_{(s)}$, $s \in \calS$, are all pairwise disjoint.

Note that condition (\ref{eq:condition-zero-error}) differs from the more usual notion of reliability where the probability of decoding error gets arbitrarily close to zero as the number of channel uses increases. Here, not only the probability of error must be exactly zero (as in zero-error information theory \cite{Korner.Orlitsky1998}), but also the channel can be used only once. While the constraint on a single channel use may seem restrictive, note that in many practical situations, in particular in network coding, a single message may already be large enough to encompass the whole communication session, so that further channel uses are not allowed.

\subsubsection{Perfect secrecy}
The scheme is said to be \emph{perfectly secret} if absolutely no information is leaked to the eavesdropper, i.e.,
\begin{equation}\label{eq:condition-perfect-secrecy}
  I(S;W) = 0.
\end{equation}
Equivalently, the uncertainty about the message is not reduced by the eavesdropper's observation.

Throughout the paper, we will use the word \emph{secure} as a synonym for \emph{secret}. We will also refer to a perfectly secure scheme simply as \emph{secure}.

Note that condition (\ref{eq:condition-perfect-secrecy}) corresponds to perfect secrecy in the Shannon sense \cite{Shannon1949:Secrecy} and is stronger than the usual notion of secrecy in the information-theoretic security literature \cite{Liang++2009}, where the average information leakage (per channel use) gets arbitrarily close to zero as the number of channel uses increases.

\subsection{Wiretap Networks}
\label{ssec:wiretap-networks}

We now consider the case where the wiretap channel is a linear network coding system potentially subject to errors. Consider an $(n \times m)_q$ linear coded network specified by matrices $C \in \Fq^{|\calE| \times n}$ and $F \in \Fq^{|\calE| \times |\calE|}$ and a set of receivers $\mathfrak{R} \subseteq \calP(\calE)$, where $\calE = \{1,\ldots,|\calE|\}$ denotes the set of network edges. The network is used to communicate a message $S \in \calS$ to each receiver, which is done by encoding $S$ into an input matrix $X \in \Fq^{n \times m}$ for transmission over the network.

As described in Section~\ref{ssec:linear-network-coding}, the output matrix at a receiver $\calR \in \mathfrak{R}$ is given by
\begin{equation}\nonumber{}
  Y(\calR) = C_\calR X + F_{\calR} Z
\end{equation}
where $Z \in \Fq^{|\calE| \times m}$ denotes the matrix of error packets. Let
\begin{equation}\nonumber
  \calZ_x \triangleq \{z \in \Fq^{|\calE| \times m} \colon P(z|x) > 0\}
\end{equation}
denote the set of all \emph{possible} values for $Z$ when the input matrix is $x$. Then the set of all possible $Y(\calR)$ given $x$ is obtained as
\begin{equation}\nonumber
  \calY_x(\calR) = \left\{y \in \Fq^{|\calR| \times m} \colon y = C_\calR x + F_{\calR} z,\, z \in \calZ_x \right\}.
\end{equation}
Since there are multiple receivers, a coding scheme for such a network consists of not only an encoder $P(X|S)$ but also a decoding function for each receiver. Accordingly, we say that the scheme is zero-error if it is zero-error for each individual receiver $\calR \in \mathfrak{R}$.

For the remainder of the paper, we will focus on the case where $Z$ has at most $t$ nonzero rows, i.e.,
\begin{equation}\nonumber
  \calZ_x \triangleq \{z \in \Fq^{|\calE| \times m} \colon \wt(z) \leq t\}
\end{equation}
where $\wt(Z)$ denotes the number of nonzero rows of $Z$. In this case, a zero-error scheme is said to be a \emph{$t$-error-$\rho$-erasure-correcting} scheme, where $\rho$ denotes the rank deficiency of the linear coded network. Note that even when $t=0$, a $0$-error-$\rho$-erasure-correcting scheme must still be able to guarantee reliable (zero-error) communication for all receivers, i.e., it must able to make up for the rank deficiency experienced by the receivers.

Suppose there is an eavesdropper who can observe the packets transmitted on a subset of links $\calI \subseteq \calE$. The corresponding matrix observed by the eavesdropper is given by
\begin{equation}\nonumber
  C_\calI X + F_{\calI} Z.
\end{equation}
Here, we assume the worst case where the eavesdropper has access to the matrix $Z$ (possibly because $Z$ was selected by the eavesdropper), so we define the eavesdropper observation as
\begin{equation}\nonumber
  W(\calI) = C_\calI X.
\end{equation}
Consider the case where the eavesdropper is allowed to \emph{arbitrarily} choose any $\calI \subseteq \calE$ with $|\calI| \leq \mu$. Since the eavesdropper may choose $\calI$ in a worst-case or adversarial fashion, this situation may be modeled mathematically by assuming that there are multiple eavesdroppers, each with one of the allowed subsets $\calI$. Accordingly, we say that the scheme is \emph{secure under $\mu$ observations} if it is perfectly secure for all $\calI$ such that $|\calI| \leq \mu$.
%Similarly, we say that the scheme is \emph{piecewise secure under $\mu$ observations and $g$ guesses} if it is piecewise secure under $g$ guesses for all $\calI$ such that $|\calI| \leq \mu$.

Under this model, $\mu$ may be viewed as a security parameter, while $t$ may be viewed as a reliability parameter.

\begin{remark}
  As observed in \cite{Rouayheb.Soljanin2007}, the type II wiretap channel of \cite{Ozarow.Wyner1984} can be viewed as the special case of a two-node network with a single receiver $\calR = \calE$, where $|\calE| = n$, $t=0$ and $C$ is an identity matrix.
\end{remark}

\subsection{Universal Schemes}

The specification of a wiretap network requires the specification of $C$, $F$ and $\mathfrak{R}$, as well as $m$, $\mu$ and $t$. As a consequence, the properties of a coding scheme designed for a network are tied to the particular network code used; there is no guarantee that the scheme will work well over other networks.

In this paper, we are interested in universal schemes, i.e., schemes that share the same property (i.e., $t$-error-$\rho$-erasure-correcting, secure under $\mu$ observations) for \emph{all} possible network codes. As we shall see, this approach not only has practical benefits but also greatly simplifies the theoretical analysis.

\medskip
\begin{definition}\label{def:universally-correcting}
  A coding scheme for an $(n \times m)_q$ linear coded network is \emph{universally $t$-error-$\rho$-erasure-correcting} if it is zero-error under the fan-out set
\begin{multline}\nonumber
  \calY_x = \left\{(A,y) \in \Fq^{n \times n} \times \Fq^{n \times m} \,\right|\, y = Ax + Z,\, \\ \left. \rank A \geq n-\rho,\, \rank Z \leq t,\, Z \in \Fq^{n \times m} \right\}.
\end{multline}
\end{definition}
\medskip

Note that, to incorporate the fact that the matrix $C_\calR$ is known at receiver $\calR$, we have to include the matrix $A$ as part of the channel output. Contrasted with the models in \cite{Silva.Kschischang2009:Metrics}, the model in Definition~\ref{def:universally-correcting} may be interpreted as a \emph{worst-case} coherent network coding channel.

\medskip
\begin{proposition}
  A universally $t$-error-$\rho$-erasure-correcting scheme for an $(n \times m)_q$ network is $t$-error-$\rho$-erasure-correcting for \emph{any} $(n \times m,\,n-\rho)_q$ network regardless of the network code or the set of receivers.
\end{proposition}
\begin{proof}
  We will show that, for any $C_\calR \in \Fq^{|\calR| \times n}$ with $\rank C_\calR \geq n-\rho$, any $F_\calR \in \Fq^{|\calE| \times |\calE|}$, and any $Z \in \Fq^{|\calE| \times m}$ with $\wt(Z) \leq t$, a receiver that knows $C_\calR$ and $Y(\calR) = C_\calR X + F_\calR Z$ can successfully decode using a universal decoder. First, since $n \geq \rank C_\calR$, there exists, regardless of $|\calR|$, some matrix $T \in \Fq^{n \times |\calR|}$ such that $\rank TC_\calR = \rank C_\calR$. Now, since $\rank TC_\calR \geq n-\rho$ and $\rank TF_\calR Z \leq t$, we have that $(TC_\calR,TY(\calR)) \in \calY_x$. Thus, the receiver can successfully decode by applying the universal decoder on $(A,Y) = (TC_\calR,TY(\calR))$.
\end{proof}
\medskip

\begin{definition}
  Consider an $(n \times m)_q$ linear coded network with input matrix $X \in \Fq^{n \times m}$. A coding scheme is \emph{universally secure under $\mu$ observations}
  if it is perfectly secure
  for each eavesdropper observation $W = BX$, for all $B \in \Fq^{\mu \times n}$.
\end{definition}
\medskip

Clearly, a universally secure scheme is always secure regardless of the network code.

Focusing on universal schemes immediately offers the analytical advantage of not having to specify the network topology and the network code, except for the parameters $n$, $m$, $q$, $\rho$. These parameters provide an interface between the problems of network code design and end-to-end code design, which then become completely independent.

Of course, the rates achieved by a universal scheme could potentially be smaller than those of non-universal schemes; equivalently, to achieve an optimal rate, a universal scheme may impose certain constraints on the interface parameters. Our goal in this paper is to determine exactly what rates are achievable by universal schemes, as well as to construct computationally efficient schemes that achieve these rates.

\section{Universal Error Correction}
\label{sec:error-correction}

We start by considering the case where $\mu = 0$, i.e., there is no eavesdropper (or security is not a concern).

Below we show a result that, while similar to the results in \cite{Silva.Kschischang2009:Metrics}, is not available there, as model considered here is slightly different.

\medskip
\begin{theorem}\label{thm:error-correction-deterministic-encoder}
  Consider a deterministic encoder $X = \enc(S)$, where $\enc\colon \calS \to \calX$, and let $\calC = \{\enc(s),\, s \in \calS\}$. Then the encoder is universally $t$-error-$\rho$-erasure-correcting if and only if $\dr(\calC) > 2t + \rho$.
\end{theorem}
\begin{proof}
  The correction guarantee has been proved in \cite{Silva.Kschischang2009:Metrics}. We now prove the converse. Suppose $\dr(\calC) = d \leq 2t + \rho$ and let $x_1,x_2 \in \calC$ be such that $x_1 \neq x_2$ and $\rank(x_2 - x_1) = d$. Let $A \in \Fq^{n \times n}$ be a matrix whose right null space is a subspace of $\linspan{x_2 - x_1}$ with dimension $\min\{\rho,d\}$. Let $E = A(x_2 - x_1)$. Then $\rank A \geq n - \rho$ and $\rank E = d - \min\{\rho,d\} = \max\{d-\rho,0\} \leq 2t$. Let $E_1,E_2 \in \Fq^{n \times m}$ be such that $E = E_1 - E_2$, $\rank E_1 \leq t$, and $\rank E_2 \leq t$. Then $y = Ax_1 + E_1 = A x_2 + E_2$, and therefore $(A,y) \in \calY_{x_1} \cap \calY_{x_2}$. Since the encoder is deterministic, $x_1$ and $x_2$ must correspond to distinct messages, which implies that the scheme is not zero-error.
\end{proof}
\medskip

Theorem~\ref{thm:error-correction-deterministic-encoder} shows that, in the case of a deterministic encoder, the correction capability of a scheme is characterized precisely by the minimum rank distance of the image of the encoder. Thus, the problem can be solved by (and only by) a rank-metric code with sufficiently large minimum rank distance. While Theorem~\ref{thm:error-correction-deterministic-encoder} is concerned only with the encoder, computationally efficient decoders (for a Gabidulin code) have been proposed in \cite{Silva++2008} (see also \cite{Silva.Kschischang2009:FastDecoding-ISIT,Silva2009(PhD)}) for all values of $\rho$ and $t$.

The theorem also shows a tradeoff between errors and erasures that is analogous to that of classical coding theory; namely, an error can be traded for two erasures, and vice-versa, with the ``exchange currency'' being the minimum rank distance of the code. 
% Note, however, that this interpretation is valid only from the perspective of the \emph{encoder}, i.e., decoders for schemes with distinct $(\rho,t)$ may be quite different in general.

It is important to note that the characterization in Theorem~\ref{thm:error-correction-deterministic-encoder} is valid only for deterministic encoding. In the case of stochastic encoding, it is conceivable that the same message $s$ could give rise to two distinct codewords $x_1$ and $x_2$ with small rank distance (so that they would be indistinguishable at the receiver), yet the \emph{message} $s$ itself could be successfully decoded. Thus, while the direct part of the theorem still holds (as long as $H(X|S) = 0$), the converse does not. Surprisingly, however, the same interplay between errors and erasures that exists for a deterministic encoder (namely, one error is equivalent to two erasures) still remains for a stochastic encoder, as shown in the next result.

\medskip
\begin{theorem}\label{thm:error-correction-stochastic-encoder}
  Consider an $(n \times m)_q$ linear coded network. An encoder that is universally $t$-error-$\rho$-erasure-correcting is also universally $t'$-error-$\rho'$-erasure-correcting for all $t',\rho'\geq 0$ such that $2t'+\rho' \leq 2t+\rho$.
\end{theorem}
\begin{proof}
  See the Appendix.
\end{proof}
\medskip

The result of Theorem~\ref{thm:error-correction-stochastic-encoder} will be crucially used in Section~\ref{ssec:noisy-converse} to prove a converse theorem for networks subject to errors and observations.

A consequence of Theorem~\ref{thm:error-correction-stochastic-encoder} is that we could safely restrict attention to encoders that are universally $\rho$-erasure-correcting (that is, universally $0$-error-$\rho$-erasure-correcting), since erasure-correction capability can be naturally traded for error-correction capability. However, as before, note that the result of Theorem~\ref{thm:error-correction-stochastic-encoder} applies only to the encoder, i.e., it is in principle not trivial to obtain a decoder for one scheme given a decoder for the other.

\section{Perfect Secrecy for Noiseless Networks}
\label{sec:noiseless}

In this section we treat the case of an $(n \times m, n)$ linear coded network subject to $\mu$ observations but no errors---so that the channel from the transmitter to each receiver is noiseless. Thus, each receiver can correctly recover the channel input $X$.

From now on, unless otherwise mentioned, we assume that the message space is $\calS = \Fq^{k \times m}$, so that the message is a $k \times m$ matrix. The rows of $S$, denoted $S_1,\ldots,S_k \in \Fq^{1 \times m}$, may then be viewed as packets. All logarithms are taken to the base $q^m$, so that information is measured in $q^m$-ary units, or \emph{packets}.

\subsection{Preliminaries}
\label{ssec:noiseless-preliminaries}

We start by reviewing the basic idea of \emph{coset coding}, which was proposed by Ozarow and Wyner for the special case of the type II wiretap channel \cite{Ozarow.Wyner1984}, and later applied to the general case by Rouayheb and Soljanin \cite{Rouayheb.Soljanin2007}. The scheme requires each packet to be an element of a finite field, i.e., $\Fq^{1 \times m}$ must be a field. In the following, we assume that $m=1$.

Let $\calC$ be an $[n,n-k]$ linear code over $\Fq$ with parity-check matrix $H \in \Fq^{k \times n}$. The transmitter encodes $S$ into $X$ by choosing uniformly at random some $X \in \Fq^n$ such that $S = HX$. In other words, each message is viewed as a syndrome specifying a coset of $\calC$, and the transmitted word is randomly chosen among the elements of that coset. Upon reception of $X$, decoding is performed by simply computing the syndrome $S = HX$. Thus, the scheme is always zero-error.

For the special case of a type II wiretap channel, it can be shown \cite{Ozarow.Wyner1984} that the scheme is perfectly secure if $\calC$ is an MDS code and $k \leq n-\mu$. In general, however, this may not be sufficient. The following result is shown in \cite{Rouayheb.Soljanin2007}.

\medskip
\begin{theorem}[\cite{Rouayheb.Soljanin2007}]\label{thm:secrecy-condition-base-field}
  In the coset coding scheme described above, assume that the eavesdropper observes $W = BX$, where $B \in \Fq^{\mu \times n}$. If $I(S;W) = 0$, then $H(S) \leq n-\mu$. Moreover, if $H(S) = k = n-\mu$, then
  \begin{equation}\label{eq:secrecy-condition-base-field}
    I(S;W) = 0 \quad \iff \quad \linspan{H} \cap \linspan{B} = 0.
  \end{equation}
\end{theorem}
\medskip

The above result can be used to design a network code based on a given parity-check matrix $H$ \cite{Rouayheb.Soljanin2007}. More precisely, the network code (i.e., the global coding matrix $C$) must be constructed in such a way that, for all $\calI$ with $|\calI| \leq \mu$, the matrix $B=C_\calI$ satisfies (\ref{eq:secrecy-condition-base-field}) for the given $H$. Note that, since there is always some $B$ violating (\ref{eq:secrecy-condition-base-field}), the scheme is not universal.

Although the case $m>1$ is not considered in \cite{Rouayheb.Soljanin2007}, it is easy to see that any scheme for $m=1$ can immediately be extended to $m>1$ by applying the scheme $m$ times in a component-wise fashion. Encoding is performed identically, by randomly choosing $X$ such that $S = HX$ (where $S$ and $X$ are now matrices), and the same holds for the decoding. Clearly, the resulting scheme retains exactly the same properties of the original one (in particular, non-universality).

\subsection{A Universal Scheme}
\label{ssec:noiseless-universal}

As we have seen above, in order to directly apply Ozarow-Wyner's coset coding scheme, packets must be elements of a finite field, and one way to achieve this is to assume $m=1$. Another approach, the one we propose in this paper, is to make use of the vector space isomorphism $\Fq^{1 \times m} \cong \Fqm$. In other words, we regard packets as elements of a finite field $\Fqm$; this is still compatible with $\Fq$-linear network coding since $\Fqm$ is a vector space over $\Fq$.

Theorem~\ref{thm:secrecy-condition-base-field} holds unchanged, provided we replace $\Fq$ with $\Fqm$ (note that we can regard $B \in \Fq^{\mu \times n}$ as a matrix over $\Fqm$, since $\Fq \subseteq \Fqm$). However, all the entries of $B$ still lie in the \mbox{\emph{subfield} $\Fq$}.
Since $\Fq \subseteq \Fqm$, we can regard $B \in \Fq^{\mu \times n}$ as a matrix over $\Fqm$. After replacing $\Fq$ with $\Fqm$, Theorem~\ref{thm:secrecy-condition-base-field} follows unchanged.

Under this interpretation, the variables $S \in \Fqm^k$, $X \in \Fqm^n$, $W \in \Fqm^\mu$ are now viewed as column vectors over $\Fqm$. For the purposes of Theorem~\ref{thm:secrecy-condition-base-field}, we can regard $B \in \Fq^{\mu \times n}$ as a matrix over $\Fqm$ (since $\Fq \subseteq \Fqm$). Then the theorem follows unchanged after replacing $\Fq$ with $\Fqm$. Note, however, that all the entries of $B$ \emph{still lie} in the \mbox{subfield $\Fq$}. As the number of possibilities for $H \in \Fqm^{k \times n}$ is now much larger as compared to $B$ (for $m>1$), it is conceivable that some $H$ exists satisfying (\ref{eq:secrecy-condition-base-field}) for all $B$. This can be seen as the crucial ingredient that enables universal security. A suitable choice of $H$ is given in the next theorem.

\medskip
\begin{theorem}\label{thm:MRD-vertical-nonsingular}
  Let $\calC$ be an $[n,n-k]$ linear code over $\Fqm$ with parity-check matrix $H \in \Fqm^{k \times n}$. If $\dr(\calC) = k+1$ and $\mu \leq n-k$, then
  \begin{equation}\label{eq:MRD-rank-condition}
    \rank \mat{H \\ B} = \rank H + \rank B, \quad \text{for all $B \in \Fq^{\mu \times n}$.}
  \end{equation}
Conversely, if $\mu = n-k$, then (\ref{eq:MRD-rank-condition}) holds only if  $\dr(\calC) = k+1$.
\end{theorem}
\begin{proof}
Suppose that, for some $\mu \leq n-k$, there exists some matrix $B \in \Fq^{\mu \times n}$ such that
\begin{equation}\nonumber
  \rank \mat{H \\ B} < \rank H + \rank B.
\end{equation}
Let $r = \rank B$, let $T \in \Fq^{r \times \mu}$ be some full-rank matrix such that $\rank TB = r$, and let $D \in \Fq^{(n-k-r) \times n}$ be some full-rank matrix such that the matrix
\begin{equation}\nonumber
  B' = \mat{TB \\ D} \in \Fq^{(n-k) \times n}
\end{equation}
is full-rank. We have that
\begin{align}
\rank \mat{H \\ B'}
&\leq \rank \mat{H \\ B} + \rank D \nonumber \\
&< \rank H + \rank B + \rank D = n. \nonumber
\end{align}
Let
\begin{equation}\nonumber
   M = \mat{H \\ B'}.
\end{equation}
Since $\rank M < n$, there must exist some nonzero $x \in \Fqm^n$ such that $Mx = 0$. But this implies that $Hx = 0$, i.e., $x \in \calC$, and $B'x = 0$, i.e., $\rank \phi_m(x) \leq k$. Thus, $\dr(\calC) \leq k$.

For the converse statement, suppose that $\dr(\calC) \leq k$. Then there exists some nonzero $x \in \calC$ such that $\rank \phi_m(x) \leq k$. This implies that there exists some full-rank $B \in \Fq^{(n-k) \times n}$ such that $Bx=0$. Note also that $Hx=0$. Thus,
\begin{equation}\nonumber
  \rank \mat{H \\ B} < n = \rank H + \rank B.
\end{equation}
\end{proof}
\medskip

In order to state our main result in full generality, we first present a generalization of Theorem~\ref{thm:secrecy-condition-base-field}.

\medskip
\begin{lemma}\label{lem:zero-info-linear}
  Let $H \in \Fqm^{k \times n}$ and $B \in \Fqm^{\mu \times n}$. Let $X \in \Fqm^n$, $S = HX$ and $W = BX$ be random variables. Let $\calS = \{Hx :\,  x \in \Fqm^n\}$ and, for all $s \in \calS$, let $\calX_s = \{x \in \Fqm^n :\, s = Hx\}$.
\begin{enumerate}
  \item \label{item1-lemma:zero-info-linear} If $X$ is uniform over    $\mathcal{X}_s$ given $S=s$, then
\begin{equation}\nonumber
  I(S;W) \leq \rank H + \rank B - \rank \mat{H \\ B}.
\end{equation}
  \item \label{item2-lemma:zero-info-linear} If $S$ is uniform over $\calS$, then
\begin{equation}\nonumber
  I(S;W) \geq \rank H + \rank B - \rank \mat{H \\ B}.
\end{equation}
\end{enumerate}
\end{lemma}
\begin{proof}
See the Appendix.
\end{proof}
\medskip

We can now state the main result of this section.

\medskip
\begin{theorem}\label{thm:universal-security}
Consider an $(n \times m,\, n)_q$ linear coded network. Let $\calC$ be an $[n,n-k]$ linear code over $\Fqm$ with parity-check matrix $H \in \Fqm^{k \times n}$. A coset coding scheme based on $H$ is universally secure under $\mu$ observations if $k \leq n-\mu$, $m \geq n$ and $\calC$ is MRD. Conversely, in the case of a uniformly distributed message, the scheme is universally secure under $n-k$ observations only if $\calC$ is an MRD code with $m \geq n$.
\end{theorem}
\begin{proof}
  The achievability follows from item 1) of Lemma~\ref{lem:zero-info-linear}, Theorem~\ref{thm:MRD-vertical-nonsingular}, and the definition of an MRD code. The partial converse follows from item 2) of Lemma~\ref{lem:zero-info-linear}, Theorem~\ref{thm:MRD-vertical-nonsingular}, and the Singleton bound (\ref{eq:singleton-linear}).
\end{proof}
\medskip

The following example illustrates the constructive part of Theorem~\ref{thm:universal-security}.

\medskip
\begin{example}
  Let $q=2$, $m=n=3$, $\mu=2$ and $k=n-\mu=1$. Let $\Fqm = \mathbb{F}_{2^3}$ be generated by a root of $p(x) = x^3 + x + 1$, which we denote by $\alpha$. According to \cite{Gabidulin1985}, one possible $[n,\mu]$ linear MRD code over $\Fqm$ has parity-check matrix $H = \mat{1 & \alpha & \alpha^2}$.

To form $X = \mat{X_1 & X_2 & X_3}^T \in \Fqm^n$ given a source message $S \in \Fqm$, we can choose $X_2,X_3 \in \Fqm$ uniformly at random and set $X_1$ to satisfy
\begin{equation}\nonumber
  S = HX = X_1 + \alpha X_2 + \alpha^2 X_3.
\end{equation}
Note that $X$ can be transmitted over any $(n \times m,\,n)_q$ linear coded network. The specific network code used is irrelevant as long as each destination node is able to recover $X$.

Now, suppose that the eavesdropper intercepts $W = BX$, where
\begin{equation}\nonumber
  B = \mat{1 & 0 & 1 \\ 0 & 1 & 1}.
\end{equation}
Then
\begin{align}
  W &= B\mat{X_1 \\ X_2 \\ X_3} = \mat{1 & 0 & 1 \\ 0 & 1 & 1}\mat{S + \alpha X_2 + \alpha^2 X_3 \\ X_2 \\ X_3} \nonumber \\
&= \mat{1 \\ 0}S +  \mat{\alpha & 1+\alpha^2 \\ 1 & 1} \mat{X_2 \\ X_3}. \nonumber
\end{align}
This is a linear system with $3$ variables and $2$ equations over $\Fqm$. Note that, given $S$, there is exactly one solution for $(X_2,X_3)$ for each value of $W$. Thus, ${\sf Pr}(W|S) = 1/8^{2}$, $\forall S,W$, from which follows that $S$ and $W$ are independent.
\end{example}

\subsection{Encoder Structure}
\label{ssec:noiseless-encoder}

In this subsection, we develop a more concrete encoder structure for the coset coding scheme proposed above.

Let $H \in \Fqm^{k \times n}$ be the parity-check matrix of an $[n,n-k]$ linear code over $\Fqm$. Let $T \in \Fqm^{n \times n}$ be an invertible matrix such that
\begin{equation}\nonumber
  T^{-1} = \mat{H \\ H_1}
\end{equation}
for some $H_1 \in \Fqm^{(n-k) \times n}$. Consider the following encoder. Given a message $S \in \Fqm^k$, the encoder chooses $V \in \Fqm^{(n-k)}$ uniformly at random and independently from $S$, and produces $X \in \Fqm^n$ by computing
\begin{equation}\nonumber
  X = T \mat{S \\ V}.
\end{equation}

\medskip
\begin{proposition}\label{prop:encoder-condition-parity}
  The encoder described above is universally secure under $\mu \leq n-k$ observations if the code defined by $H$ is MRD with $m \geq n$. %The converse holds if $\mu = n-k$ and the message is uniformly distributed.
\end{proposition}
\begin{proof}
  Note that $S = HX$ and $V = H_1 X$. Then Theorem~\ref{thm:universal-security} holds if we can prove that $X$ is uniform given $S$, i.e., if $H(X|S) = n-k$. By expanding $H(V,X|S)$ in two ways, we have
  \begin{align}
  H(X|S)
  &= H(V|S) + H(X|V,S) - H(V|X,S) \nonumber \\
  &= H(V|S) + H(X|V,S) \nonumber \\
  &= H(V|S) \nonumber \\
  &= H(V) \nonumber \\
  &= n-k. \nonumber
  \end{align}
\end{proof}
\medskip

Note that this equivalence between the two encoders has been previously shown in \cite{Rouayheb.Soljanin2007} for the case of $m=1$ with non-universal security (i.e., when $H$ satisfies the conditions of Theorem~\ref{thm:secrecy-condition-base-field} for a specific network).

We now give a security condition based directly on the matrix $T$ rather than its inverse.

\medskip
\begin{proposition}\label{prop:encoder-condition-direct}
  The encoder described above is universally secure under $\mu \leq n-k$ observations if the last $n-k$ rows of $T^T$ form a generator matrix of an $[n,n-k]$ linear MRD code over $\Fqm$ with $m \geq n$. %The converse holds if $\mu = n-k$ and the message is uniformly distributed.
\end{proposition}
\begin{proof}
  Let $G \in \Fqm^{(n-k) \times n}$ and $G_1 \in \Fqm^{k \times n}$ be such that
  \begin{equation}\nonumber
    T^T = \mat{G_1 \\ G}.
  \end{equation}
  Then
  \begin{equation}\nonumber
    \mat{I & 0 \\ 0 & I} = T^{-1} T = \mat{H \\ H_1} \mat{G_1^T & G^T} = \mat{H G_1^T & H G^T \\ H_1 G_1^T & H_1 G^T}.
  \end{equation}
  Thus, $H G^T = 0$. Since both $G$ and $H$ are full-rank, it follows that $G$ and $H$ are generator and parity-check matrices, respectively, for exactly the same code.
\end{proof}

\subsection{Converse Results}
\label{ssec:noiseless-converse}

We now prove that our scheme is optimal with respect to packet length, i.e., the scheme minimizes the required packet length among all universal schemes. For generality, in the following theorem we revert to the notation of Section~\ref{sec:problem-formulation} (with matrices over the base field $\Fq$).

%We start with a very simple auxiliary result.
%
%\medskip
%\begin{lemma}\label{lem:entropy-lower-bound}
%  $H(S) \geq -\log \max p(S)$.
%\end{lemma}
%\begin{proof}
%  By Jensen's inequality,
%  \begin{equation}\nonumber
%    E\left[-\log \frac{p(S)}{\max p(S)}\right] \geq -\log E \left[ \frac{p(S)}{\max p(S)} \right] \geq 0.
%  \end{equation}
%\end{proof}

\medskip
\begin{theorem}\label{thm:converse-noiseless}
Consider a noiseless $(n \times m,\, n)_q$ linear coded network. Assume that the source message has entropy of $k$ packets. There exists a zero-error scheme that is universally secure under $\mu = n-k$ observations only if $m \geq n$.
\end{theorem}
\begin{proof}
By assumption of zero-error communication (and of a noiseless network), there is a function $f\colon \Fq^{n \times m} \to \calS$ such that $S = f(X)$. Thus, we may write $\calX_s = \{x \in \Fq^{n \times m}: f(x) = s\}$.
Now,
\begin{align}
k
&= H(S) \nonumber \\
&= H(S|X,W) + I(S;X,W) \nonumber \\
&= I(S;X,W) \label{eq:converse-noiseless-proof-1} \\
&= I(S;W) + I(S;X|W) \nonumber \\
&= I(S;X|W) \label{eq:converse-noiseless-proof-2} \\
&= H(X|W) - H(X|S,W) \nonumber \\
&\leq H(X|W) \label{eq:converse-noiseless-proof-3} \\
&\leq n-\rank B. \label{eq:converse-noiseless-proof-4}
\end{align}
where (\ref{eq:converse-noiseless-proof-1}) follows since $S$ is a function of $X$ and (\ref{eq:converse-noiseless-proof-2}) follows since $I(S;W) = 0$. Since (\ref{eq:converse-noiseless-proof-4}) holds with equality for all full-rank $B \in \Fq^{\mu \times n}$, we must have $H(X|S,W) = 0$ and $H(X|W) = n-\mu$ for all such $B$. By the chain rule of entropy, it is not hard to see that the latter condition implies that $X$ is uniform (for instance, by choosing each $B$ as a submatrix of an identity matrix, as in the wiretap channel II). Thus,  $H(X) = n$. Since $H(X) = H(X,S) = H(S) + H(X|S)$, we have that $H(X|S) \geq n - k = \mu$. Thus, there must be some $s^* \in \calS$ such that $H(X|S=s^*) \geq \mu$, which implies that $|\calX_{s^*}| \geq q^{m\mu}$.
%Thus, $p(S) = |\calX_S|/q^{mn}$. Since, by Lemma~\ref{lem:entropy-lower-bound}, $\max p(S) \geq q^{-m H(S)} = q^{-mk}$, it follows that there exists some $s^* \in \calS$ such that $|\calX_{s^*}| \geq q^{m\mu}$.
On the other hand, the fact that $H(X|S,W) = 0$ for all full-rank $B$ implies that $X$ must be uniquely determined given $W = BX$ and the indication that $X \in \calX_S$. From Theorem~\ref{thm:MRD-vertical-nonsingular}, this implies that each $\mathcal{X}_s$ must be a rank-metric code with $\dr(\mathcal{X}_s) \geq n-\mu + 1$. In particular, $\dr(\calX_{s^*}) \geq n-\mu+1$. From the Singleton bound (\ref{eq:singleton-bound}), we see that this can only happen if $m\geq n$.
\end{proof}
\medskip

As Theorem~\ref{thm:converse-noiseless} shows, if $m < n$, universal schemes do not exist. For $m \geq n$, not only do universal schemes exist, but also they achieve exactly the same rates as the best non-universal schemes. It should be noted, however, that these results assume the requirements of perfect secrecy and zero-error communication. If these conditions are relaxed to asymptotically perfect secrecy and vanishing error probability (over multiple channel uses), then it is possible to construct universal schemes even for $m=1$ \cite{Cheraghchi++2009}.

\section{Perfect Secrecy for Noisy Networks}
\label{sec:noisy}

In this section, we treat the general case of an $(n \times m, n-\rho)_q$ linear coded network subject to $t$ errors and $\mu$ observations.

\subsection{A Universal Scheme}
\label{ssec:noisy-universal}

Consider the encoder described on Section~\ref{ssec:noiseless-encoder}. Assume that the input variable for the encoder is $S' \in \Fqm^{(n-\mu)}$ (rather than $S$). Given $S'$, the encoder produces
\begin{equation}\nonumber
  X = T \mat{S' \\ V}
\end{equation}
where $V \in \Fqm^{\mu}$ is chosen uniformly at random and independently from $S'$, and $T \in \Fqm^{n \times n}$ is an invertible matrix.
%Suppose that $m \geq n$ and that the last $\mu$ rows of $T^T$, denoted by $\Gin \in \Fqm^{\mu \times n}$, form a generator matrix for an $[n,\mu]$ linear MRD code over $\Fqm$. Then it follows from Proposition~\ref{prop:encoder-condition-direct} that the scheme is universally secure under $\mu$ observations, \emph{regardless} of the distribution of $S'$.
Suppose that $m \geq n$, and let $\Gin \in \Fqm^{\mu \times n}$ denote the last $\mu$ rows of $T^T$. It follows from Proposition~\ref{prop:encoder-condition-direct} that, if $\Gin$ is a generator matrix for an $[n,\mu]$ linear MRD code over $\Fqm$, then the scheme is universally secure under $\mu$ observations, \emph{regardless} of the distribution of $S'$.

Suppose we choose
\begin{equation}\nonumber
  S' = \mat{0 \\ S}
\end{equation}
where $S \in \Fqm^k$ is the ``true'' message, and $k \leq n-\mu$. Then the scheme remains universally secure under $\mu$ observations. On the other hand, the redundancy in $S'$ may be useful to provide error correction.

Let us define an auxiliary variable
\begin{equation}\nonumber
  U = \mat{S \\ V}.
\end{equation}
Then the encoder effectively maps $U$ into $X$ via the deterministic mapping
\begin{equation}\nonumber
  X = G^T U
\end{equation}
where $G \in \Fqm^{(k+u) \times n}$ denotes the last $k+\mu$ rows of $T^T$. In particular, the set of all possible $X$ is given by
\begin{equation}\nonumber
  \calC = \{G^T u,\, u \in \Fqm^{(k+\mu)}\}.
\end{equation}
Then, it follows from Theorem~\ref{thm:error-correction-deterministic-encoder} that, when $X$ is transmitted over an $(n\times m,\,n-\rho)_q$ network subject to $t$ errors, the receiver can uniquely determine $U$ (and therefore $S$) if $\dr(\calC) > 2t+\rho$. This condition is satisfied if $\calC$ is an $[n,k+\mu]$ linear MRD code over $\Fqm$ and $k+\mu \leq n - (2t + \rho)$.

The above analysis proves the following result.

\medskip
\begin{theorem}\label{thm:universal-secrecy-error-control}
Consider an $(n \times m)_q$ linear coded network. In the encoder described above, assume that $G \in \Fqm^{(k+\mu) \times n}$ is the generator matrix of an $[n,k+\mu]$ linear MRD code over $\Fqm$ such that the last $\mu$ rows of $G$ form a generator matrix of an $[n,\mu]$ linear MRD code over $\Fqm$. The scheme is universally $t$-error-$\rho$-erasure-correcting and universally secure under $\mu$ observations if $m \geq n$ and $k \leq n - 2t - \rho - \mu$.
\end{theorem}
\medskip

Whenever an error control encoder satisfies the secrecy conditions of Theorem~\ref{thm:universal-secrecy-error-control}, we will say it is \emph{secrecy-compatible}.

As mentioned in Section~\ref{sec:error-correction}, decoding can be performed using the methods in \cite{Silva++2008,Silva2009(PhD)} if $\calC$ is a Gabidulin code. In this case, if $G$ is given in the form (\ref{eq:generator-gabidulin}), then it is easy to see that any $\mu$ consecutive rows of $G$ (in particular the last ones) indeed form a generator matrix of an MRD sub-code.

\subsection{Converse Results}
\label{ssec:noisy-converse}

In this section, we prove that our proposed scheme is optimal, both in the sense of achieving the maximum possible rate and in the sense of requiring the minimum possible packet length among all schemes that achieve this maximum rate. As in Theorem~\ref{thm:converse-noiseless}, we use the generic notation of Section~\ref{sec:preliminaries}.

\medskip
\begin{theorem}\label{thm:converse-noisy}
Consider an $(n \times m)_q$ linear coded network. Assume that the source message has entropy of $k$ packets. There exists a scheme that is universally $t$-error-$\rho$-erasure-correcting and universally secure under $\mu$ observations only if $k \leq n-2t-\rho-\mu$. Moreover, this maximum rate can be attained only if $m \geq n$.
\end{theorem}
\begin{proof}
Let $n' = n-2t-\rho$. Let $B \in \Fq^{\mu \times n}$ be a full-rank matrix and let $A \in \Fq^{n' \times n}$ be a full-rank matrix such that $B = PA$ for some (necessarily full-rank) $P \in \Fq^{\mu \times n'}$. Let $Y_A = AX$ and $W_B = BX = PY_A$. If the encoder is universally $t$-error-$\rho$-erasure-correcting then, by Theorem~\ref{thm:error-correction-stochastic-encoder}, it is also universally $(2t+\rho)$-erasure-correcting. Thus, %for all full-rank $A \in \Fq^{n' \times n}$,
there is a function $f_A\colon \Fq^{n' \times m} \to \calS$ such that $S = f_A(Y_A)$. In particular, there is also a function $f \colon \Fq^{n \times m} \to \calS$ such that $S = f(X)$. Thus, we may write $\calX_s = \{x \in \Fq^{n \times m}: f(x) = s\}$.
Now,
\begin{align}
k
&= H(S) \nonumber \\
&= H(S|Y_A,W_B) + I(S;Y_A,W_B) \nonumber \\
&= I(S;Y_A,W_B) \label{eq:converse-noisy-proof-1} \\
&= I(S;W_B) + I(S;Y_A|W_B) \nonumber \\
&= I(S;Y_A|W_B) \label{eq:converse-noisy-proof-2} \\
&= H(Y_A|W_B) - H(Y_A|S,W_B) \nonumber \\
&\leq H(Y_A|W_B) \label{eq:converse-noisy-proof-3} \\
&\leq n'-\rank P = n' - \mu \label{eq:converse-noisy-proof-4}
\end{align}
where (\ref{eq:converse-noisy-proof-1}) follows since $S$ is a function of $Y_A$ and (\ref{eq:converse-noisy-proof-2}) follows since $I(S;W_B) = 0$.
This proves the first statement. Now consider the second statement. Since (\ref{eq:converse-noisy-proof-4}) holds with equality, we must have $H(Y_A|S,W_B) = 0$ and $H(Y_A|W_B) = n'-\mu$. Note that these conditions hold for all full-rank $B$ and all $A \in \calA_B$, where
\begin{equation}\nonumber
  \calA_B = \{A \in \Fq^{n' \times n} : \rank A = n',\, \linspan{B} \subseteq \linspan{A}\}.
\end{equation}
This implies that $H(\{Y_A: A \in \calA_B\} | S,W_B) = 0$ and therefore $H(\bar{Y}_B | S,W_B) = 0$, where $\bar{Y}_B = \bar{A}_B X$ and $\bar{A}_B$ is the matrix consisting of the vertical stacking of all matrices in $\calA_B$. It is not hard to see that, as long as $n' > \mu$, $\rank \bar{A}_B = n$. (In fact, $\bar{A}_B$ contains every nonzero vector of $\Fq^{1 \times n}$ as one of its rows.) It follows that $H(X|S,W_B) = 0$, for all full-rank $B$. Thus, $X$ must be uniquely determined given $W_B = BX$ and the indication that $X \in \calX_S$. From Theorem~\ref{thm:error-correction-deterministic-encoder}, this implies that each $\mathcal{X}_s$ must be a rank-metric code with $\dr(\mathcal{X}_s) \geq n-\mu + 1$.

On the other hand, we have seen that, for each full-rank $A \in \Fq^{n' \times n}$, it holds that $H(Y_A|W_B) = n'-\mu$ for \emph{all} full-rank $P \in \Fq^{\mu \times n'}$, where $W_B = P Y_A$ and $B=PA$. By the chain rule of entropy, it is not hard to see that this implies that $Y_A$ is uniform (for instance, by choosing some $P$'s that are submatrices of an identity matrix, as in the wiretap channel~II). Thus, $H(Y_A) = n'$, which implies that $H(X) \geq n'$. Since $H(X) = H(X,S) = H(S) + H(X|S)$, we have that $H(X|S) \geq n' - k = \mu$. Thus, there must be some $s \in \calS$ such that $H(X|S=s) \geq \mu$, which implies that $|\calX_s| \geq q^{m\mu}$. Together with the fact that $\dr(\mathcal{X}_s) \geq n-\mu + 1$, we can see, from the Singleton bound (\ref{eq:singleton-bound}), that this can only happen if $m\geq n$.
\end{proof}

\section{Practical Considerations}
\label{sec:practical-considerations}

\subsection{Packet Length}

The schemes proposed in this paper all require that the packet length $m$ be at least as large as the batch size (i.e., the number of transmitted packets) $n$. This is the only constraint imposed by universal schemes---in sharp contrast with previous approaches that require the network code to be known and field size $q$ to be significantly large.
In practice, the requirement on the packet length is usually easily satisfied: typical random network coding implementations use $m \gg n$, for instance, $m\geq 1024$ (with $q=256$) while $n \leq 256$ \cite{Wang.Li2007:R2,Chou++2003}.

\subsection{Layered Structure}

The fact that a single encoder/decoder pair simultaneously provides both secrecy and error control offers a great deal of simplicity and flexibility to the proposed scheme. A block diagram of the scheme is illustrated in Fig.~\ref{fig:layers}.
\begin{figure*}
  \centering
  \includegraphics[scale=0.95]{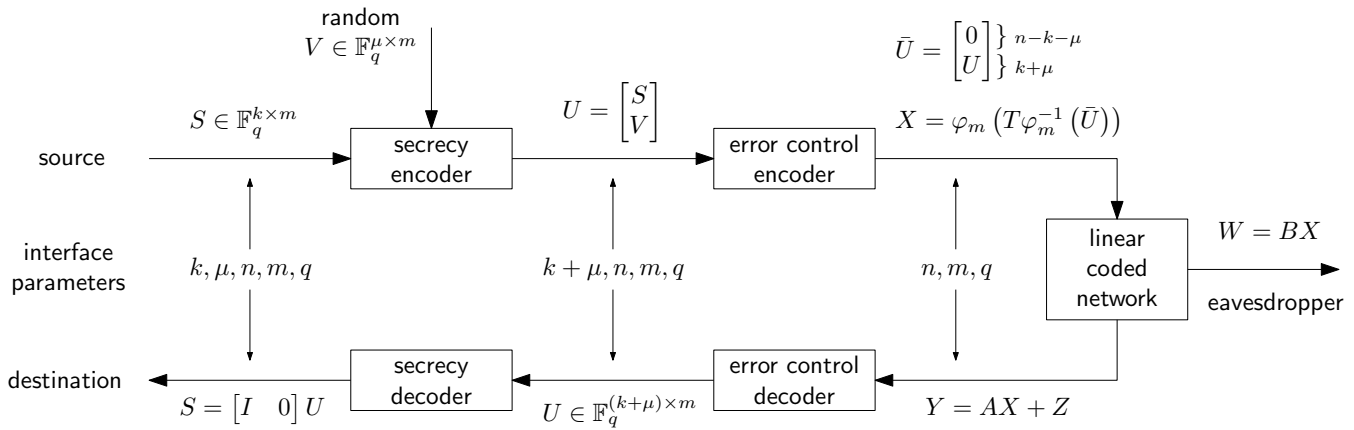}\\
  \caption{Layer structure of the proposed coding scheme, incorporating both secrecy and reliability. The source message is guaranteed to be secret from the eavesdropper provided $\rank B \leq \mu$. The destination is guaranteed to reliably recover the message provided that $\rank Z \leq (\rank A - k - \mu)/2$.}
  \label{fig:layers}
\end{figure*}
We can view the system as consisting of three layers. The first layer accepts a message of $k$ packets and performs secrecy coding simply by concatenating the message with $\mu$ random packets. The second layer accepts the secrecy-encoded message and applies secrecy-compatible error control coding. (For clarity, the isomorphism between $\Fqm^n$ and $\Fq^{n \times m}$ is shown explicitly in Fig.~\ref{fig:layers}.) The resulting codeword, consisting of $n$ packets, is then delivered to the third layer, which corresponds to the linear coded network. The interface parameters are $k$, $\mu$, $n$ and $m$, where $n-k-\mu$ determines the amount of error control (note that $d=n-k-\mu+1$ is the minimum rank distance of the code). The matrix $T \in \Fqm^{n \times n}$ in Fig.~\ref{fig:layers} is such that, for all $i=1,\ldots,n$, the last $i$ rows of $T^T$ form a generator matrix of an $[n,i]$ linear MRD code over $\Fqm$. Provided that the error control decoder associated with $T$ is flexible to handle any amount of error control given as an input parameter (this is possible for the decoders in \cite{Silva++2008,Silva2009(PhD)}), we obtain a scheme that is ``universal'' in yet another sense: the same scheme can be used regardless of the parameters $\mu$, $t$ and $\rho$ (assuming $n > \mu+2t+\rho$). As we can see from Fig.~\ref{fig:layers}, we can easily trade among rate, secrecy and error control by simply adjusting the interface parameters $k$ and $\mu$.

\subsection{Cartesian Products of Codes}

According to the structure described in Sections~\ref{ssec:noiseless-encoder} and~\ref{ssec:noisy-universal}, encoding and decoding of a source message are performed via matrix-by-vector multiplication with arithmetic over an extension field. Specifically, encoding and decoding can be performed with, respectively, $O(k'n)$ and $O(n^2)$ arithmetic operations in $\Fqm$, where $k' = k+\mu$. For moderate to large $m$, these operations may turn out to be quite expensive since, in practice, a multiplication in $\Fqm$ costs about $m^2$ operations in $\Fq$.

A convenient way to reduce this complexity is to use Cartesian products of MRD codes. Assume that $m = rn$, for some $r$. We construct a code in $\Fq^{n \times nr}$ via the isomorphism $\phi_n^{(n \times r)}$, rather than $\phi_{nr}^{(n \times 1)}$ as before. Let $\calC^r \subseteq \Fqn^{n \times r}$ denote the $r$-fold Cartesian product of a code $\calC \subseteq \Fqn^{n}$ with itself. Suppose $\calC$ is defined by the generator and parity-check matrices $G \in \Fqn^{k' \times n}$ and $H \in \Fqn^{(n-k') \times n}$, i.e., $\calC = \{G^Tu: u \in \Fqn^{k'}\} = \{x \in \Fqn^n: Hx = 0\}$. Then it follows that $\calC^r = \{G^Tu: u \in \Fqn^{k' \times r}\} = \{x \in \Fqn^{n \times r}: Hx = 0\}$. It is also clear that $\dr(\phi_n(\calC^r)) = \dr(\calC)$. Thus, all the results and methods of this paper can be equally applied to $\calC^r$. In particular, decoding is performed by applying a decoder for $\calC$ column-wise on the received matrix $X \in \Fqn^{n \times r}$. As a consequence, the encoding and decoding complexity are reduced to, respectively, $O(k'nr) = O(k'm)$ and $O(n^2r) = O(nm)$ operations in the smaller field $\Fqn$.

\subsection{Using Low-Complexity Normal Bases}

The encoding and decoding complexity can be reduced even further by using a normal basis to perform extension field arithmetic.

Let $\alpha \in \Fqn$. If the elements $\alpha, \alpha^{q^1}, \ldots, \alpha^{q^{n-1}}$ are linearly independent over $\Fq$, then $\{\alpha^{q^0},\ldots,\alpha^{q^{n-1}}\}$ is called a \emph{normal basis} for $\Fqn$ over $\Fq$, and $\alpha$ is called a \emph{normal element}. Suppose the matrix $T \in \Fqn^{n \times n}$ is given by $T = [T_{ij}]$ where $T_{i,j} = \alpha^{[i-1+j-1]}$, for $1 \leq i,j \leq n$. Then $T$ not only is invertible, but also satisfies both requirements of secrecy and error control, as any contiguous subset of rows of $T$ is a generator matrix of an MRD code \cite{Gabidulin1985}. Now, if the basis generated by $\alpha$ is also used to implement the arithmetic over $\Fqn$, then significant complexity savings can be obtained, as described in \cite{Silva.Kschischang2009:FastDecoding-ISIT,Silva2009(PhD)}. Specifically, suppose that $q$ is a power of 2 and that $\alpha$ is a self-dual, optimal normal element constructed via Gauss periods \cite{Gao1993:Thesis,Gao++2000:AlgorithmsExponentiation}. Then decoding can be performed with approximately $5(n-k')^2nm + \frac{1}{2}n^2 m$ multiplications and $10(n-k')^2nm + \frac{1}{2}n^2 m$ additions in $\Fq$, while encoding can be performed with just $2k'nm$ additions (XORs) in $\Fq$ \cite{Silva2009(PhD)}. Note that, if error control is not used (i.e., $k' = n$), then the decoding complexity is smaller than performing Gaussian elimination on the received matrix, and the encoding complexity is even much smaller.

Although normal bases exist over any finite field, normal bases satisfying the above requirements exist only for certain choices of the extension degree $n$. In particular, for $q=256$, the choices of $n$ are limited to $n=3$, 5, 9, 11, 23, 29, 33, 35, 39, 41, 51, 53, 65, 69, 81, 83, 89, 95, $99, \ldots$ \cite{Silva.Kschischang2009:FastDecoding-ISIT}. As can be seen, there is still a reasonable degree of flexibility that should be suitable for most applications. On the other hand, if low-complexity (though not necessarily optimal) normal bases are used (while retaining the properties of self-duality and Gaussianity), then an even greater degree of flexibility is possible (although with a slightly increased complexity).

\subsection{Extension to Noncoherent Network Coding}

The scheme described in the paper is suitable for coherent network coding and is indeed optimal. In the case of noncoherent (random) network coding, the scheme can be adapted by including appropriate packet headers. More precisely, the transmission matrix should be $\mat{I & X}$, where $X$ is the transmission matrix of the original scheme. Clearly, including packet headers does not affect security (since the only information carried by the headers is the coding vectors, which are already assumed to be known by the eavesdropper), but allows the scheme to be decoded when the transfer matrix $A$ is unknown. It is shown in \cite{Silva++2008} that such adaptation preserves the error-correcting capability of the code, so the universally $t$-error-$\rho$-erasure-correcting property is maintained. Although the rate achieved in this case is no longer optimal, it is very close to optimal for all practical packet lengths \cite{Silva++2008}.

%\subsection{Batch Size}
%
%The fact that our proposed scheme is based on an MRD code allows us significant flexibility in selecting the batch size without compromising the decoding complexity. Consider a scheme $\mathfrak{C}_1$ that is universally $t$-error-$\rho$-erasure-correcting over an $(n \times m)_q$ linear coded network, and suppose we want to find an efficient decoder for $\mathfrak{C}_1$. Suppose we are given an efficient decoder $\textsf{D}'(\cdot)$ for some scheme that is universally $t$-error-$\rho'$-erasure-correcting over an $(n' \times m)_q$ network.
%
%
%Consider two schemes $\mathfrak{C}$ and $\mathfrak{C}'$: the first is universally $t$-error-$\rho$-erasure-correcting over an $(n \times m)_q$ linear coded network, while the second is universally $t$-error-$\rho'$-erasure-correcting over an $(n' \times m)_q$ linear coded network, where $\rho' = \rho+n'-n$ and $n \leq n' \leq m$.
%
%
%Suppose we want to construct a decoder for a $t$-error-$\rho$-erasure-correcting scheme for an $(n \times m)_q$ linear coded network, but we are only given an efficient decoder for a $t$-error-$\rho'$-erasure-correcting scheme for an $(n' \times m)_q$ linear coded network, where $\rho' = \rho+n'-n$ and $n \leq n' \leq m$.
%
%
%In certain practical scenarios (as will be seen in Section~?), we may know efficient decoders only for certain values of $n$.
%
%Explain that not all $m$'s are good for gaussian bases, but we can work it out by using an $m$ slightly larger than $n$.

\section{Conclusion}
\label{sec:conclusion}

In this paper, we have addressed the problem of achieving secure and reliable communication over a linear coded network subject to wiretapping and also possibly to jamming. We have shown that \emph{universal} schemes exist if the packet length is sufficiently large. In this case, no coordination is needed between the designs of the outer code and of the underlying network code; in particular, the field size for the network code may be chosen as the minimum required for feasibility. We have also shown that our proposed scheme is optimal in the sense of achieving the maximum possible rate and requiring the minimum possible packet length among all schemes that achieve this maximum rate. The proposed scheme is flexible in that it defines two layers above the network coding layer: a secrecy layer and a (secrecy-compatible) error control layer. The amount of information rate, secrecy protection and error control provided by the scheme can be easily traded off against each other simply by adjusting the interface parameters.

The main tool that we use in this paper is the theory of rank-metric codes. The proposed scheme borrows from our previous work on error control for network coding (without secrecy constraints) and admits very efficient encoding and decoding.

For a network that transports $n$ packets with rank deficiency~$\rho$, and is under the threat of an adversary who can eavesdrop on $\mu$ links and inject $t$ error packets, we have shown that the maximum achievable rate is at most $n - \rho - 2t - \mu$. This result assumes perfect secrecy and (one-shot) zero-error communication. If the latter requirement is relaxed to vanishingly small error probability, then it is possible to achieve a higher rate of $n - \rho - t -\mu$, provided that both the field size and the packet length grow to infinity. A natural, yet unsolved question is how to achieve this higher rate without requiring the field size to grow. Such a solution, if one exists, would reassure the ``separation principle'' advocated by this paper: that basic network coding on the on hand, and secrecy/error control protection on the other hand, can be treated as belonging to completely independent layers.

Another possible avenue for future work might be to generalize the results of this paper beyond multicast problems. An initial step in this direction has been given in \cite{Khisti++2010:ISIT-SecureBroadcast}.

\section*{Acknowledgements}

The authors would like to thank the Associate Editor and the
anonymous reviewers for their many helpful comments, which have resulted in a substantially improved paper.

\bibliographystyle{IEEEtran}
\bibliography{IEEEabrv,networkcoding,codingtheory,rankmetric,silva,finitefields}

\appendix
\label{sec:appendix}

\begin{proof}[Proof of Theorem~\ref{thm:error-correction-stochastic-encoder}]
  Let $\calY_x$ and $\calY_x'$ denote the fan-out sets of Definition~\ref{def:universally-correcting} for $(\rho,t)$ and $(\rho',t')$, respectively.
  We have to prove that if $s_1,s_2 \in \calS$ are distinct messages such that $(A',y') \in \calY_{x_1}' \cap \calY_{x_2}'$ for some $x_1 \in \calX_{s_1}$ and $x_2 \in \calX_{s_2}$, then the sets $\calY_{(s)}$, $s \in \calS$, (given by (\ref{eq:fan-out-overall})) are not all pairwise disjoint.

  Write $y' = A' x_1 + E_1' = A'x_2 + E_2'$, where $\rank A' \geq n - \rho'$, $\rank E_1' \leq t$ and $\rank E_2' \leq t'$. Let $E' = E_1' - E_2' = A'(x_2 - x_1)$, and note that $\rank E' \leq 2t'$.

  First, consider the case where $t' - t = \Delta > 0$.  Let $T \in \Fq^{n \times n}$ be a matrix whose right null space is a subspace of $\linspan{E'}$ with dimension $\min\{2\Delta,\rank E'\}$. Let $E = TE'$ and $A = TA'$. Then $\rank T \leq n - 2\Delta$,
  \begin{align}
  \rank E
  &= \rank E' - \min\{2\Delta,\rank E'\} \nonumber \\
  &\leq \max\{2t' -2\Delta,0\} \leq 2t \nonumber
  \end{align}
  and
  \begin{align}
  \rank A
  &\geq \rank T + \rank A' - n \nonumber \\
  &\geq n-2\Delta + n - \rho' - n \geq n - \rho. \nonumber
  \end{align}
%  \begin{equation}\nonumber
%    \rank E = \rank E' - \min\{2\Delta,\rank E'\} \leq \max\{2t' -2\Delta,0\} \leq 2t
%  \end{equation}
%  and
%  \begin{equation}\nonumber
%    \rank A \geq \rank T + \rank A' - n \geq n-2\Delta + n - \rho' - n \geq n - \rho.
%  \end{equation}
  Let $E_1,E_2 \in \Fq^{n \times m}$ be such that $E = E_1 - E_2$, $\rank E_1 \leq t$, $\rank E_2 \leq t$. Then $y = Ax_1 + E_1 = A x_2 + E_2$, and therefore $(A,y) \in \calY_{x_1} \cap \calY_{x_2} \subseteq \calY_{(s_1)} \cap \calY_{(s_2)}$.

  Now, consider the case where $\rho' - \rho = 2\Delta > 0$. Let $R \in \Fq^{n \times n}$ and $A = A' + R$ be such that $\rank R = 2\Delta$ and $\rank A = \rank A' + \rank R$. Then $\rank A \geq n-\rho' + 2\Delta = n - \rho$. Let $E = E' + R(x_2 - x_1) = A(x_2 - x_1)$. Note that $\rank E \leq \rank E' + \rank R \leq 2t' + 2\Delta \leq 2t$. Once again, let $E_1,E_2 \in \Fq^{n \times m}$ be such that $E = E_1 - E_2$, $\rank E_1 \leq t$, and $\rank E_2 \leq t$. Then $y = Ax_1 + E_1 = A x_2 + E_2$, and therefore $(A,y) \in \calY_{x_1} \cap \calY_{x_2} \subseteq \calY_{(s_1)} \cap \calY_{(s_2)}$.

  The case where both $t' \leq t$ and $\rho' \leq \rho$ follows immediately from Definition~\ref{def:universally-correcting}.
\end{proof}

\medskip
\begin{proof}[Proof of Lemma~\ref{lem:zero-info-linear}]
  To prove the first statement, let $\mathcal{W} = \{Bx : x \in \Fqm^n\}$ and
\begin{equation}\nonumber
  \mathcal{X}_{s,w} = \left\{x \in \Fqm^n : \mat{s \\ w} = \mat{H \\ B}x\right\}.
\end{equation}
Observe that
\begin{align}
H(W) &\leq \log_{q^m} |\mathcal{W}| = \rank B \nonumber \\
H(X|S) &= \log_{q^m} |\mathcal{X}_S| = n- \rank H \nonumber \\
H(X|S,W) &\leq \log_{q^m} |\mathcal{X}_{S,W}| = n- \rank \mat{H \\ B}. \nonumber \end{align}
By expanding $I(S,X;W)$ and noting that $W$ is a function of $X$, we have
\begin{align}
I(S;W)
&= I(S,X;W) - I(X;W|S) \nonumber \\
%&= I(S,X;W) - H(X|S) + H(X|S,W) \nonumber \\
&= H(W) - H(X|S) + H(X|S,W) \nonumber \\
&\leq \rank B + \rank H - \rank \mat{H \\ B}. \nonumber
\end{align}

To prove the second statement, first note that
\begin{equation}\nonumber
  \dim (\rowsp{H} \cap \rowsp{B}) = \rank \mat{H \\ B} - \rank H - \rank B
\end{equation}
where $\rowsp{\cdot}$ denotes the row space of a matrix. Let $t = \dim (\rowsp{H} \cap \rowsp{B})$. Then there exist full-rank matrices $T_1 \in \Fqm^{t \times \mu}$ and $T_2 \in \Fqm^{t \times k}$ such that $T_1 B = T_2 H$ and $\rank T_2 H = t$. This implies that
\begin{equation}\nonumber
  T_1 W = T_1 B X = T_2 H X = T_2 S.
\end{equation}
Since $S$ is uniform, we have that $I(S;W) \geq H(T_2 S) = t$.
\end{proof}

\end{document}